\documentclass[journal,twoside,web,xcdraw,svgnames]{ieeecolor}
\usepackage{generic}
\usepackage{cite}
\usepackage{amsmath,amssymb,amsfonts}

\usepackage{hyperref}

\usepackage[T1]{fontenc}
\usepackage[utf8]{inputenc}
\usepackage{graphicx}
\usepackage{amsmath,systeme,amssymb,amsfonts}
\usepackage[version=4]{mhchem}
\let\labelindent\relax
\usepackage{siunitx}
\usepackage{enumitem}
\usepackage{longtable,tabularx}
\usepackage{comment}
\usepackage{textgreek}
\usepackage{float}
\usepackage[utf8]{inputenc}
\usepackage{tikz, pgfplots}
\usepackage{textcomp}
\usepackage{tcolorbox} 
\usepackage{graphicx}
\usepackage{amsmath}

\usepackage[version=4]{mhchem}
\usepackage{siunitx}
\usepackage{longtable,tabularx}
\usetikzlibrary{positioning, shapes, arrows}
\usepackage{thmtools}
\usepackage{textcomp}
\usepackage{soul}
\usepackage{float,amsfonts,amsthm,color,xcolor}
\usepackage{amssymb}
\usepackage{graphics} 
\usepackage{epstopdf}
\usepackage{booktabs} 
\usepackage{caption, subcaption}
\usepackage{breqn}
\usepackage{array, mathtools}
\usepackage{pgfplots}
\usepackage{siunitx}
\usepackage{algorithm}
\usepackage[noend]{algpseudocode}
\usepackage{tabulary}
\usetikzlibrary{arrows.meta}
\usepackage{stmaryrd} 
\usepackage{varwidth}
\usepackage{comment}
\usetikzlibrary{positioning}

\usepackage{amsthm} \theoremstyle{definition} \newtheorem{theorem}{Theorem} \newtheorem{definition}{Definition} \newtheorem{assumption}{Assumption}  \newtheorem{lemma}{Lemma} \newtheorem{remark}{Remark}  \newtheorem{example}{Example}

\def\BibTeX{{\rm B\kern-.05em{\sc i\kern-.025em b}\kern-.08em
    T\kern-.1667em\lower.7ex\hbox{E}\kern-.125emX}}
\begin{document}
\title{Stabilization Limits of Payoff-Based Higher-Order Replicator Dynamics}
\author{Hassan Abdelraouf, Vijay Gupta, and Jeff S. Shamma 
\thanks{H. Abdelraouf and V. Gupta are with the School of Electrical and Computer Engineering, Purdue University, USA. J. S. Shamma is with the Department of Industrial and Enterprise Systems Engineering, University of Illinois at Urbana-Champaign,  USA. (e-mails: abdelra5@purdue.edu, gupta869@purdue.edu, jshamma@illinois.edu) }}
\maketitle

\begin{abstract}
Replicator dynamics (RD) is a fundamental model in learning in games, connecting evolutionary game theory and online learning. This paper studies payoff-based higher-order variants of RD represented as a cascade interconnection between an integrator in parallel with an auxiliary linear time-invariant (LTI) system and the softmax mapping. We investigate learnability of Nash equilibria under this Nash-stationary learning rule.  First, we revisit recent results that establish convergence to Nash Equilibrium whenever the auxiliary LTI system is strictly passive and prove  a converse passivity result: if the auxiliary LTI system is not passive, then there exists a static strictly contractive game whose interior Nash equilibrium is unstable under the closed-loop learning dynamics. Second, we show that there exists a class of games with isolated interior Nash equilibria that cannot be locally asymptotically stabilized by any payoff-based higher-order RD whose auxiliary LTI system is asymptotically stable and strictly proper. Finally, we show that if Nash stationarity (i.e., all Nash equilibria are stationary points of the learning dynamics) is relaxed, then generalized exponential RD (Ex-RD) can locally asymptotically stabilize a logit equilibrium for any continuously differentiable game. The stabilized equilibrium  can be viewed as an entropy-regularized approximate Nash equilibrium.
\end{abstract}

\section{Introduction} \label{sec:introduction}

Population games \cite{sandholm2010population,hofbauer1998evolutionary} provide a framework for modeling strategic interactions in large populations, where the population state describes the fraction of agents using each strategy and the payoff vector assigns payoffs to these strategies. Learning dynamics, or evolutionary dynamics, describe how the population state changes in response to the payoff signal. From a systems perspective, the game maps strategies to payoffs, while the learning dynamics define a map from payoffs to strategies. Therefore, the resulting learning process can be viewed as a closed-loop feedback interconnection between the game and the learning model \cite{fox2013population,arcak2020dissipativity,park2018passivity,abdelraouf2025passivity}. In particular, passivity theory has proven to be an important tool to analyze stability and convergence to Nash equilibria~\cite{abdelraouf2025passivity}, with notions such as \(\delta\)-passivity used in stable games \cite{fox2013population,park2019population}, equilibrium-independent passivity used in zero-sum games \cite{gao2020passivity}, and counterclockwise dissipativity used in potential games~\cite{martins2025counterclockwise}. 

A widely studied learning rule in population games is the replicator dynamics (RD). Standard replicator dynamics is a first-order dynamic since the payoff signal drives a single score variable through a pure integrator, with no additional internal state. Higher-order learning dynamics introduce auxiliary states driven either by payoffs \cite{arslan2006anticipatory,laraki2013higher,hankins2025nash} or by strategies \cite{gao2023second}. These auxiliary states can incorporate prediction or memory to improve performance \cite{abdelraouf2026can} and enhance stability properties \cite{toonsi2023higher,arslan2006anticipatory}. In this paper, we focus on payoff-based higher-order variants of replicator dynamics (RD), where the payoff signal is passed through an integrator in parallel with an additional linear time-invariant (LTI) system. The resulting payoff-to-score transfer function has the form \(G(s)I_n\), where \( G(s)=\frac{1}{s}+g(s), \) and the strategy is obtained through the softmax map. Equivalently, the learning dynamics can be viewed as a cascade interconnection between \((1/s+g(s))I_n\) and the softmax mapping.

Recent work~\cite{abdelraouf2026convergence} showed an important connection of passivity of $g(s)$ and the convergence of payoff-based higher order RD. Specifically, if  \(g(s)\) is strictly passive, then payoff-based higher-order RD locally stabilizes an isolated interior Nash equilibrium in any contractive game. Here, we are interested in two questions to complete the relation between passivity and stability of payoff-based higher order RD. First, for contractive games, is passivity of $g(s)$ also necessary for stability? This question was partially answered in~\cite{mabrok2016passivity}, which showed that if in the absence of \(\delta\)-passivity, a {\em dynamic} contractive
game can be constructed that destabilizes the feedback interconnection. However, whether passivity suffices for {\em static} contractive games remains unknown. Second, even for non-contractive games, is passivity (or a related property) of \(g(s)\) sufficient to at least locally stabilize any isolated interior Nash equilibrium?

The main contributions of this paper are threefold. First, we establish a converse to the sufficiency result in~\cite{abdelraouf2026convergence}: while {\em strict passivity} of \(g(s)\) guarantees local stability of an isolated interior Nash equilibrium in contractive games, passivity is necessary for universal local stability over this class. Specifically, if \(g(s)\) is not passive, then there exists a {\em static} strictly contractive game whose interior Nash equilibrium is unstable under the closed-loop learning dynamics;  (Section~\ref{sec:passivity_necessary}). Second, we show that there exists a class of (non-contractive) games with isolated interior Nash equilibria that cannot be locally asymptotically stabilized by any payoff-based higher-order RD whose payoff-to-score transfer function has the form \(G(s)=1/s+g(s)\), where \(g(s)\) is asymptotically stable and strictly proper; (Section~\ref{sec:stabilization_limit}). Third, we show that this limitation can be bypassed if Nash stationarity is relaxed. In particular, Ex-RD admits a locally asymptotically stable logit equilibrium for any continuously differentiable game. This equilibrium is generally not Nash, but it can be interpreted as an entropy-regularized approximation of a Nash equilibrium; (Section~\ref{sec:beyond_nash_stationarity}).

\noindent\textbf{Notation:}
Let \(\mathbb R_{\ge 0}:=[0,\infty)\) and \(\mathbb R_{>0}:=(0,\infty)\). For
\(x\in\mathbb R^n\), \(x_i\) denotes its \(i\)th component,
\(\operatorname{diag}(x)\) denotes the diagonal matrix with diagonal entries
\(x_i\), and \(\|x\|\) denotes the Euclidean norm. For a matrix \(A\),
\(\|A\|_2\) denotes the induced Euclidean norm. The identity matrix is denoted
by \(I_n\), and \(\mathbf 1_n\) and \(\mathbf 0_n\) denote the vectors of all
ones and all zeros, respectively. For matrices \(A\) and \(B\),
\(A\otimes B\) denotes their Kronecker product, and
\[
    \operatorname{blkdiag}(A,B)
    =
    \begin{bmatrix}
        A & 0\\
        0 & B
    \end{bmatrix}
\]
denotes their block-diagonal concatenation. For a symmetric matrix \(A\),
\(A\succeq0\) and \(A\succ0\) denote positive semidefiniteness and positive
definiteness, respectively.

The probability simplex is
\(
    \Delta_n
    :=
    \{x\in\mathbb R_{\ge 0}^n:\mathbf 1_n^\top x=1\},
\)
with interior \(\operatorname{Int}(\Delta_n)\). Its tangent subspace is
\(
    \mathcal Z
    :=
    \{z\in\mathbb R^n:\mathbf 1_n^\top z=0\}.
\)
Throughout, \(N\in\mathbb R^{n\times(n-1)}\) denotes a matrix whose columns form
an orthonormal basis of \(\mathcal Z\); hence
\(
    N^\top N=I_{n-1},
    \;
    N^\top\mathbf 1_n=\mathbf 0_{n-1} .
\)
The softmax mapping \(\sigma:\mathbb R^n\to\operatorname{Int}(\Delta_n)\) is
defined by
\[
    \sigma(v)_i
    =
    \frac{e^{v_i}}{\sum_{j=1}^n e^{v_j}},
    \qquad i=1,\ldots,n .
\]
Its Jacobian is
\(
    \nabla\sigma(v)
    =
    \operatorname{diag}(\sigma(v))
    -
    \sigma(v)\sigma(v)^\top .
\)
Moreover, \(\nabla\sigma(v)\succeq0\),
\(
    \ker(\nabla\sigma(v))=\operatorname{span}\{\mathbf 1_n\},
    \;
    \operatorname{rank}(\nabla\sigma(v))=n-1,
\)
and
\(
    z^\top\nabla\sigma(v)z>0,
    \;
    \forall z\in\mathcal Z\setminus\{0\}.
\)
Thus
\(
    \widetilde Q(v):=N^\top\nabla\sigma(v)N
\)
is positive definite. Finally, \(\sigma\) is invariant under translations along
\(\mathbf 1_n\):
\(
    \sigma(v+c\mathbf 1_n)=\sigma(v),
    \;
    \forall v\in\mathbb R^n,\ c\in\mathbb R .
\)

\section{Preliminaries}

\paragraph{Passivity} For an asymptotically stable LTI system
\(
    \dot x=Ax+Bu,
    \;
    y=Cx+Du,
\)
with transfer function matrix
\(
    H(s)=C(sI-A)^{-1}B+D,
\)
passivity is equivalent to positive realness \cite{khalil2002nonlinear}:
\(
    H(j\omega)+H(j\omega)^*\succeq0,
    \;
    \forall \omega\in\mathbb R .
\)
Throughout the paper, passivity is used only for scalar LTI system 
\(g(s)\). Thus, \(g(s)\) is passive if
\(
    \operatorname{Re}\{g(j\omega)\}\ge0,
    \;
    \forall \omega\in\mathbb R .
\)

\paragraph{Population Games and Nash Stationarity}
We consider single-population games \cite{sandholm2010population} with
population state \(x\in\Delta_n\) and payoff mapping
\(
    p=\mathcal G(x),
    \;
    \mathcal G:\Delta_n\to\mathbb R^n .
\)
A state \(x^*\in\Delta_n\) is a Nash equilibrium (NE) of \(\mathcal G\) if
\[
    y^\top\mathcal G(x^*) - {x^*}^{\top}\mathcal G(x^*)
    \le 0,
    \qquad
    \forall y\in\Delta_n .
\]
The set of Nash equilibria is denoted by \(\operatorname{NE}(\mathcal G)\).
For \(\varepsilon\ge0\), a state \(x\in\Delta_n\) is an
\(\varepsilon\)-Nash equilibrium if
\[
    \max_{y\in\Delta_n}
    \left(
        y^\top\mathcal G(x)-x^\top\mathcal G(x)
    \right)
    \le \varepsilon .
\]
The quantity on the left hand side is called the Nash gap. Thus,
\(x\in\operatorname{NE}(\mathcal G)\) if and only if its Nash gap is zero.
In particular, if \(x^*\in\operatorname{Int}(\Delta_n)\), then
\(x^*\in\operatorname{NE}(\mathcal G)\) if and only if
\(
    \mathcal G(x^*)=\alpha\mathbf 1_n
\)
for some \(\alpha\in\mathbb R\) \cite{sandholm2010population}.

\begin{definition}[Contractive games \cite{hofbauer2009stable}]
\label{def:contractive_game}
A population game \(\mathcal G\) is contractive if
\(
    (x-y)^\top\big(\mathcal G(x)-\mathcal G(y)\big)\le0,
    \;
    \forall x,y\in\Delta_n .
\)
It is strictly contractive if equality holds only when \(x=y\). If
\(\mathcal G\) is continuously differentiable, then \(\mathcal G\) is
contractive if and only if
\(
    z^\top\nabla\mathcal G(x)z\le0,
    \;
    \forall x\in\Delta_n,\; z\in\mathcal Z .
\)
\end{definition}
\noindent A strictly contractive game has a unique NE~\cite{sandholm2015population}.

\begin{definition}[Nash stationarity \cite{sandholm2010population}]
\label{def:nash_stationarity}
A learning dynamic model written in strategy form
\(
    \dot x=\mathcal V(x,\mathcal G(x))
\)
satisfies Nash stationarity on \(\operatorname{Int}(\Delta_n)\) if, for every
\(x^*\in\operatorname{Int}(\Delta_n)\),
\[
    \mathcal V(x^*,\mathcal G(x^*))=0
    \quad \Longleftrightarrow \quad
    x^*\in\operatorname{NE}(\mathcal G).
\]
\end{definition}
\paragraph{Replicator dynamics and its higher-order variants}
Replicator dynamics (RD) is a widely studied learning rule that satisfies Nash stationarity on
\(\operatorname{Int}(\Delta_n)\) \cite{sandholm2009pairwise}. It can be viewed as a continuous-time analogue of
multiplicative weights update (MWU) \cite{abdelraouf2025passivity}, and represented by
\begin{equation}
\label{eq:RD}
    \dot z=p,
    \qquad
    x=\sigma(z),
\end{equation}
where \(p\in\mathbb R^n\) is the payoff vector, \(z\in\mathbb R^n\) is the
accumulated payoff score, and \(x\in\Delta_n\) is the induced mixed strategy. Thus, RD can be viewed as the cascade interconnection of the integrator \((1/s)I_n\) and the softmax map \(\sigma\).  Payoff-based higher-order variants of RD are obtained by augmenting the score dynamics with an additional LTI system driven by the payoff signal \cite{abdelraouf2026convergence}. Specifically, consider dynamics of the form \begin{equation} \label{eq:predictive_RD} 
\begin{aligned}
\dot r &= p,\\
\dot \xi &= (A_g\otimes I_n)\xi+(B_g\otimes I_n)p,\\
m &= (C_g\otimes I_n)\xi,\\
z &= r+m,\\
x &= \sigma(z), 
\end{aligned} \end{equation}
where \(r\in\mathbb R^n\) is the RD score, \(\xi\in\mathbb R^{mn}\) is the state of the added LTI system, and \(m\in\mathbb R^n\) is its output. The matrices \(A_g\in\mathbb R^{m\times m}\), \(B_g\in\mathbb R^{m\times1}\), and \(C_g\in\mathbb R^{1\times m}\) define the transfer function \( g(s)=C_g(sI_m-A_g)^{-1}B_g . \) Equivalently, \eqref{eq:predictive_RD} is the cascade interconnection of \( G(s)=\left(\frac{1}{s}+g(s)\right)I_n \) and the softmax map, as shown in Fig.~\ref{fig:higher_order_RD}.

 \begin{figure}[H]
     \centering
     \includegraphics[width=0.5\linewidth]{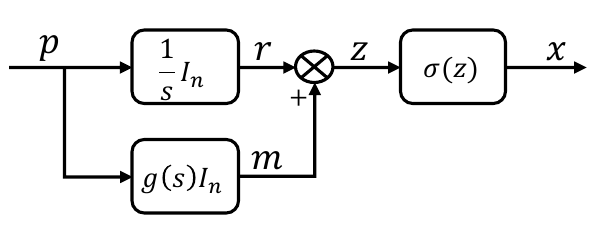}
     \caption{Representation of payoff-based higher-order RD \eqref{eq:predictive_RD}}
     \label{fig:higher_order_RD}
 \end{figure}
The
induced strategy dynamics under the payoff-based higher-order RD
\eqref{eq:predictive_RD} satisfy Nash stationarity on
\(\operatorname{Int}(\Delta_n)\) if the added LTI system characterized by
\(g(s)\) has a minimal realization and is asymptotically stable
\cite[Lemma~1]{abdelraouf2026convergence}.  Accordingly, we make the following assumption throughout, which additionally imposes strict properness to exclude direct feedthrough and hence avoid algebraic loops when the learning dynamics
are interconnected with a static game. 
\begin{assumption}
\label{ass:filter}
The payoff filter \(g(s)\) is real-rational and strictly proper, and admits a
minimal realization
\(
    g(s)=C_g(sI_m-A_g)^{-1}B_g,
\)
where \(A_g\) is Hurwitz.
\end{assumption}

\section{Passivity as a Necessary Condition for Universal Local Stability in Contractive Games}\label{sec:passivity_necessary}
For contractive games,~\cite{abdelraouf2026convergence} showed that  strict passivity of \(g(s)\) guarantees local asymptotic stability of an isolated interior Nash equilibrium with a payoff-based higher-order RD. Our first result establishes a necessity statement. Specifically, while strict passivity is sufficient for local stability, passivity is necessary for universal local stability over contractive games.
\begin{theorem} \label{thm:nonpassive_destabilization} Consider the payoff-based higher-order RD dynamics \eqref{eq:predictive_RD} under Assumption~\ref{ass:filter}. Suppose that \(g(s)\) is not passive. Then, for every \(n\ge 3\), there exists a strictly contractive matrix game \(\mathcal{G}(x)= F x\) with an isolated interior Nash equilibrium \(x^*\in\operatorname{Int}(\Delta_n)\) such that the linearization of the closed-loop payoff-based higher-order RD dynamics around the equilibrium corresponding to \(x^*\) has an eigenvalue in the open right-half plane. Consequently, \(x^*\) is unstable. \end{theorem}

\begin{proof}
Let
\(
    G(s)=\frac{1}{s}+g(s).
\)
Since \(g(s)\) is not passive, there exists \(\omega_0\neq0\) such that
\(
    \operatorname{Re}\{g(j\omega_0)\}<0.
\)
Because
\(
    \operatorname{Re}\left\{\frac{1}{j\omega_0}\right\}=0,
\)
we have
\(
    \operatorname{Re}\{G(j\omega_0)\}
    =
    \operatorname{Re}\{g(j\omega_0)\}<0.
\)
Moreover, \(\omega_0\neq0\), and since \(g(s)\) is asymptotically stable, \(G(s)\) has no pole at \(j\omega_0\). Thus \(G\) is continuous in a neighborhood of \(j\omega_0\). Hence, for sufficiently small \(\varepsilon>0\), the point \( s_0=\varepsilon+j\omega_0 \) satisfies \( \operatorname{Re}(s_0)>0, \; \operatorname{Re}\{G(s_0)\}<0 . \)
Define
\(
    \mu:=\frac{1}{G(s_0)}.
\)
Since \(\operatorname{Re}\{G(s_0)\}<0\), we have \(G(s_0)\neq0\). Writing
\(\mu=-a+jb\), it follows that
\[
    a
    =
    -\operatorname{Re}(\mu)
    =
    -\frac{\operatorname{Re}\{G(s_0)\}}{|G(s_0)|^2}
    >0.
\]

We now construct a strictly contractive matrix game whose isolated interior
Nash equilibrium is \(x^*=\frac{1}{n}\mathbf 1_n\). Let
\(N\in\mathbb R^{n\times(n-1)}\) be an orthonormal basis of the tangent space
of \(\Delta_n\). At \(x^*\), the softmax Jacobian is
\[
    Q^*
    =
    \operatorname{diag}(x^*)-x^*{x^*}^\top
    =
    \frac{1}{n}
    \left(
        I_n-\frac{1}{n}\mathbf 1_n\mathbf 1_n^\top
    \right).
\]
Therefore, its restriction to the tangent space is
\(
    \widetilde Q
    :=
    N^\top Q^*N
    =
    \frac{1}{n}I_{n-1}.
\)
Since \(n\ge3\), the tangent space has dimension at least two. Define
\[
    M=\operatorname{blkdiag} \left( \begin{bmatrix}
        -a & -b\\
        b & -a
    \end{bmatrix} , -I_{n-3}\right)
\]
where the second block is absent when \(n=3\). Then \(M\) has eigenvalues
\(
    \mu, \bar\mu, -1,\ldots,-1,
\)
and
\[
    \frac{M+M^\top}{2}
    = \operatorname{blkdiag}\left(-aI_2,-I_{n-3}\right)
    \prec0, \; \text{because } a>0.
\]
Define the reduced game matrix
\(
    \widetilde F:=nM,
\) and lift it to the original coordinates by
\(
    F:=N\widetilde F N^\top.
\)
Consider the matrix game
\(
    \mathcal G(x)=Fx.
\)
Since \(N^\top\mathbf 1_n=\mathbf{0}_{n-1}\), we have
\(
    F\mathbf 1_n=\mathbf{0}_n.
\)
Hence
\(
    \mathcal G(x^*)=Fx^*=\mathbf{0}_n,
\)
so all pure strategies have equal payoff at \(x^*\). Therefore, \(x^*\) is an
interior Nash equilibrium.

We next show that the game is strictly contractive. For any
\(z\in\mathcal Z\), there exists \(\eta\in\mathbb R^{n-1}\) such that
\(z=N\eta\). Thus
\[
    z^\top Fz
    =
    \eta^\top \widetilde F\eta
    =
    n\eta^\top M\eta
    =
    n\eta^\top\frac{M+M^\top}{2}\eta.
\]
Since \((M+M^\top)/2\prec0\), it follows that
\(
    z^\top Fz<0,
    \quad
    \forall z\in\mathcal Z\setminus\{0\}.
\)
Therefore, \(\mathcal G(x)=Fx\) is strictly contractive. In particular,
\(x^*\) is isolated. Moreover,
\(
    N^\top \nabla \mathcal G(x^*)N
    =
    N^\top FN
    =
    \widetilde F
    =
    nM,
\)
which is nonsingular.

It remains to show that the corresponding equilibrium is unstable. Since \(\mathcal G(x^*)=Fx^*=\mathbf{0}_n\), an equilibrium of the closed-loop learning
dynamics corresponding to \(x^*\) is
\(
    r^*=\mathbf 0_n,
    \xi^*=\mathbf 0_{mn},
    m^*=\mathbf 0_n,
    x^*=\sigma(r^*)=\frac1n\mathbf 1_n .
\) Linearizing the closed-loop dynamics
around this equilibrium gives
\begin{equation}\label{eq:lineaized_dynamics}
    \begin{aligned}
    \dot{\Delta r} &= \Delta p,\\
    \dot{\Delta \xi}
        &=
        (A_g\otimes I_n)\Delta \xi
        +
        (B_g\otimes I_n)\Delta p,\\
    \Delta m
        &=
        (C_g\otimes I_n)\Delta \xi,\\
    \Delta x
        &=
        Q^*(\Delta r+\Delta m),\\
    \Delta p
        &=
        F\Delta x .
\end{aligned}
\end{equation}

Since \(x(t)\in\Delta_n\), every strategy perturbation satisfies
\(\Delta x\in\mathcal Z\). Moreover, since
\(
    \mathbf 1_n^\top F
    =
    \mathbf 1_n^\top N\widetilde F N^\top
    =
    \mathbf 0_n^\top ,
\)
we have \(\Delta p=F\Delta x\in\mathcal Z\). Hence the induced
strategy-payoff linearization evolves on the tangent space \(\mathcal Z\).

Let
\(
    \mathcal N:=I_m\otimes N.
\)
Projecting the linearization onto \(\mathcal Z\), write 
\(
    \Delta r=N\delta r,
    \Delta x=N\delta x,
    \Delta p=N\delta p,
    \Delta m=N\delta m,
    \Delta \xi=\mathcal N\delta \xi .
\)
Premultiplying the linearized equations by \(N^\top\) and
\(\mathcal N^\top\), respectively, gives
\begin{equation}\label{eq:projected_dynamics}
    \begin{aligned}
    \dot{\delta r} &= \delta p,\\
    \dot{\delta \xi}
        &=
        (A_g\otimes I_{n-1})\delta \xi
        +
        (B_g\otimes I_{n-1})\delta p,\\
    \delta m
        &=
        (C_g\otimes I_{n-1})\delta \xi,\\
    \delta x
        &=
        \widetilde Q(\delta r+\delta m),\\
    \delta p
        &=
        \widetilde F\delta x .
\end{aligned}
\end{equation}

Therefore, in the frequency domain,
\(
    \delta r=\frac{1}{s}\delta p, 
    \delta m=g(s)\delta p.
\)
Thus
\(
    \delta x
    =
    \widetilde Q
    \left(
        \frac{1}{s}+g(s)
    \right)\delta p
    =
    \widetilde Q G(s)\delta p.
\)
Using \(\delta p=\widetilde F\delta x\), we obtain
\(
    \delta x
    =
    \widetilde Q G(s)\widetilde F\,\delta x.
\)
Equivalently, the reduced characteristic equation is
\(
    \det\!\left(I_{n-1}-G(s)\widetilde Q\widetilde F\right)=0.
\)
By construction,
\(
    \widetilde Q\widetilde F
    =
    \frac{1}{n}I_{n-1}\cdot nM
    =
    M.
\)
Since \(M\) has eigenvalue \(\mu\), the characteristic equation contains the
factor
\(
    1-\mu G(s)=0.
\)
At \(s=s_0\),
\[
    1-\mu G(s_0)
    =
    1-\frac{1}{G(s_0)}G(s_0)
    =
    0.
\]
Thus \(s_0\) is a closed-loop pole of the reduced linearized dynamics. Since \(\operatorname{Re}(s_0)>0\), the linearization has a mode in the open right-half plane. Therefore,
by Lyapunov's indirect method, the equilibrium corresponding to \(x^*\) is
unstable.
\end{proof}

\begin{example} \label{ex:unstable_F}
Consider the payoff-based higher-order RD \eqref{eq:predictive_RD} with
\(
    g(s)=\frac{1}{(s+1)(s+2)},
    \;
    G(s)=\frac{1}{s}+g(s).
\)
Then
\(
    g(j\omega)=\frac{1}{2-\omega^2+j3\omega},
\)
and hence
\(
    \operatorname{Re}\{g(j\omega)\}
    =
    \frac{2-\omega^2}{(2-\omega^2)^2+9\omega^2}.
\)
Thus, \(g(s)\) is not passive since
\(
    \operatorname{Re}\{g(j\omega)\}<0,
    \;
    \forall |\omega|>\sqrt{2}.
\)

Choose
\(
    s_0=0.05+j2.
\)
Then
\(
    G(s_0)\approx -0.0317-0.6478j,
\)
and therefore
\(
    \mu:=\frac{1}{G(s_0)}
    \approx
    -0.0752+1.5400j.
\)
Thus, writing \(\mu=-a+jb\), we have
\(
    a=0.0752,
    \;
    b=1.5400.
\)
Let
\[
    M=
    \begin{bmatrix}
        -a & -b\\
        b & -a
    \end{bmatrix}.
\]
For \(n=3\), take \(\widetilde F=3M\), and use the tangent-space basis
\begin{equation}
\label{eq:N_projection}
    N=
    \begin{bmatrix}
        \frac{1}{\sqrt2} & \frac{1}{\sqrt6}\\[1mm]
        -\frac{1}{\sqrt2} & \frac{1}{\sqrt6}\\[1mm]
        0 & -\frac{2}{\sqrt6}
    \end{bmatrix}.
\end{equation}
Define the matrix game
\(
    \mathcal G(x)=Fx,
    \;
    F=N\widetilde F N^\top .
\)
This gives
\[
    F\approx
    \begin{bmatrix}
        -0.1505 & -2.5920 &  2.7425\\
         2.7425 & -0.1505 & -2.5920\\
        -2.5920 &  2.7425 & -0.1505
    \end{bmatrix}.
\]
The game is strictly contractive and satisfies
\(
    F\mathbf 1_3=\mathbf 0_3.
\)
Hence
\(
    x^*=\frac{1}{3}\mathbf 1_3
\)
is its unique NE. By
Theorem~\ref{thm:nonpassive_destabilization}, the equilibrium \(x^*\) is
unstable under the payoff-based higher-order RD with
\(g(s)=1/((s+1)(s+2))\) as illustrated in
Fig.~\ref{fig:unstable_F_example_1}.
\end{example}
\begin{figure}[H]
    \centering
    \includegraphics[width=0.5\linewidth]{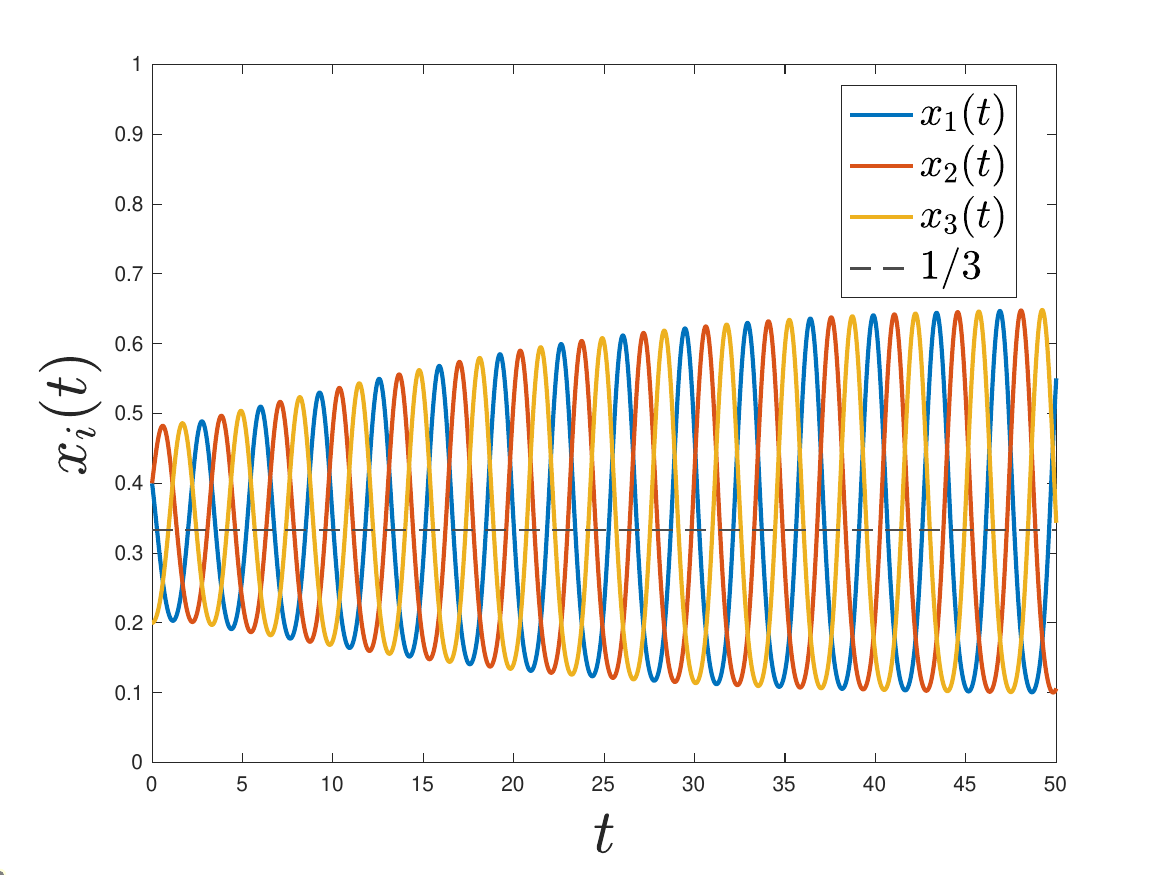}
    \caption{Implementation of payoff-based higher-order RD with $g(s)= \frac{1}{(s+1)(s+2)}$ on the game introduced in Example \ref{ex:unstable_F}}
    \label{fig:unstable_F_example_1}
\end{figure}
\begin{remark}
The construction in Theorem~\ref{thm:nonpassive_destabilization} requires
\(n\ge 3\) because the tangent space \(\mathcal Z\) must have dimension at least two in order to realize the complex conjugate pair \(\mu,\bar\mu\). When \(n=2\), the reduced dynamics are scalar. In particular, for a strictly contractive matrix game,
\(
    \widetilde F=N^\top F N=-k, \;  k>0,
\)
and since \(\widetilde Q>0\), the reduced characteristic equation takes the form
\(
    1+K G(s)=0, \;  K>0.
\)
For the example 
\[
    G(s)=\frac{1}{s}+\frac{1}{(s+1)(s+2)},
\]
this gives
\(
    s(s+1)(s+2)+K\big((s+1)(s+2)+s\big)=0,
\)
or equivalently
\(
    s^3+(3+K)s^2+(2+4K)s+2K=0.
\)
For every \(K>0\), the Routh--Hurwitz conditions for this cubic are satisfied  and all roots lie in the open left-half plane for every \(K>0\). Therefore, for this particular nonpassive filter, interconnection with a static strictly contractive two-strategy matrix game cannot make the interior NE unstable. 
\end{remark}

\section{Stabilization Limitations of Payoff-Based Higher-Order RD}\label{sec:stabilization_limit}

The previous section showed that passivity of the payoff filter \(g(s)\) is
necessary for universal local stability over the class of contractive games. In this section, we ask whether strictly passive, and hence internally stable, payoff filters suffice for non-contractive games as well. This question is related to the strong-stabilization limitation in
\cite{toonsi2023higher}, which shows that stabilizing certain mixed-strategy
Nash equilibria with higher-order gradient play may require internally
unstable learning dynamics. Here, we impose internal stability of the auxiliary LTI system and show a complementary limitation: there exist games whose isolated interior Nash equilibria cannot be locally asymptotically stabilized by any learning dynamics of the form \eqref{eq:predictive_RD}, regardless of the choice of the asymptotically stable strictly proper  \(g(s)\).

\begin{theorem} \label{thm:positive_eig_unstabilizable} 
Let \(x^*\in\operatorname{Int}(\Delta_n)\) be an isolated interior Nash equilibrium of a continuously differentiable game \(\mathcal G:\Delta_n\to\mathbb R^n\). Define \( Q^* := \operatorname{diag}(x^*)-x^*{x^*}^{\top}, \; \widetilde Q:=N^\top Q^*N, \; \widetilde F:=N^\top\nabla\mathcal G(x^*)N . \) Suppose that \(\widetilde Q\widetilde F\) has an eigenvalue \(\lambda\in\mathbb R_{>0}\). Then \(x^*\) cannot be locally asymptotically stabilized by any payoff-based higher-order RD \eqref{eq:predictive_RD} whose payoff filter \(g(s)\) satisfies Assumption~\ref{ass:filter}. \end{theorem}

\begin{proof}

For the payoff-based higher-order RD \eqref{eq:predictive_RD}, let \(g(s)\)
satisfy Assumption~\ref{ass:filter}, and define
\(
    G(s)=\frac{1}{s}+g(s).
\) Since \(x^*\in\operatorname{Int}(\Delta_n)\) is a Nash equilibrium, all pure strategies have equal payoff at \(x^*\). Hence
\(
    \mathcal G(x^*)=\alpha\mathbf 1_n
\)
for some \(\alpha\in\mathbb R\).
We first describe a closed-loop trajectory whose induced strategy is stationary
at \(x^*\). Choose \(r_0\in\mathbb R^n\)
such that
\(
    \sigma(r_0)=x^* .
\)
Since \(A_g\) is Hurwitz, it is nonsingular. For the constant payoff input
\(p^*=\alpha\mathbf 1_n\), define
\(
    r^*(t)=r_0+\alpha t\,\mathbf 1_n,
    \; 
    \xi^*=-\alpha(A_g^{-1}B_g)\otimes\mathbf 1_n .
\)
Then
\(
    (A_g\otimes I_n)\xi^*+(B_g\otimes I_n)p^*=0,
\)
and
\(
    m^*=(C_g\otimes I_n)\xi^*
    =
    -\alpha C_gA_g^{-1}B_g\,\mathbf 1_n .
\)
Therefore,
\(
    z^*(t)=r^*(t)+m^*
\)
differs from \(r_0\) only by a time-varying multiple of \(\mathbf 1_n\). By
the shift-invariance of the softmax map,
\(
    \sigma(z^*(t))=\sigma(r_0)=x^*,
    \; \forall t\ge0 .
\)
Thus the closed-loop dynamics admit a trajectory whose induced strategy is
stationary at \(x^*\).

Linearizing the closed-loop dynamics around this trajectory gives
\eqref{eq:lineaized_dynamics}, with
\(    F=\nabla\mathcal G(x^*),
    \;
    Q^*=\operatorname{diag}(x^*)-x^*{x^*}^{\top}.
\)
Projecting the linearized dynamics onto the tangent space of \(\Delta_n\) gives
\eqref{eq:projected_dynamics}, where
\(
    \widetilde Q=N^\top Q^*N,
    \;
    \widetilde F=N^\top\nabla\mathcal G(x^*)N .
\)
Taking the frequency-domain representation of \eqref{eq:projected_dynamics}
yields
\(
    \delta x
    =
    \widetilde QG(s)\widetilde F\delta x .
\)
Therefore, the reduced characteristic equation is
\begin{equation}
\label{eq:reduced_characteristic}
    \det\!\left(I_{n-1}-G(s)\widetilde Q\widetilde F\right)=0 .
\end{equation}

    By assumption, \(\widetilde Q\widetilde F\) has an eigenvalue
\(\lambda\in\mathbb R_{>0}\). Hence
\eqref{eq:reduced_characteristic} contains the scalar factor
\(
    1-\lambda G(s)=0 .
\)
This equation has a positive real root for every admissible
\(g(s)\). Since \(g(s)\) is real-rational, strictly proper, and asymptotically
stable, write
\(
    g(s)=\nu(s)/d(s),
\)
where \(d(s)\) is monic and Hurwitz, and \(\deg \nu<\deg d\). Then
\(
    1-\lambda G(s)=0
\)
is equivalent to
\[
    1-\lambda\left(\frac{1}{s}+\frac{\nu(s)}{d(s)}\right)=0 .
\]
Multiplying by \(sd(s)\) gives
\(
    \phi(s):=(s-\lambda)d(s)-\lambda s\nu(s)=0 .
\)
Since \(d(s)\) is a real monic Hurwitz polynomial, \(d(0)>0\). Therefore,
\(
    \phi(0)=-\lambda d(0)<0 .
\)
Moreover, since \(g(s)\) is strictly proper, the leading term of
\(\phi(s)\) comes from \(sd(s)\) and has positive coefficient. Hence
\(
    \phi(s)\to+\infty
    \; \text{as } s\to+\infty .
\)
By continuity, there exists \(\rho>0\) such that
\(
    \phi(\rho)=0 .
\)
Since \(\rho>0\) and \(d\) is Hurwitz, \(d(\rho)\neq0\). Thus
\(
    1-\lambda G(\rho)=0 .
\)

Consequently, the reduced linearization has a closed-loop pole
\(\rho\in\mathbb R_{>0}\). Therefore, the linearized closed-loop dynamics have
an eigenvalue in the open right-half plane. By Lyapunov's indirect method, the
equilibrium corresponding to \(x^*\) is unstable for every payoff filter
\(g(s)\) satisfying Assumption~\ref{ass:filter}. Consequently, \(x^*\) cannot
be locally asymptotically stabilized by any such payoff-based higher-order RD.
\end{proof}

\begin{example}
\label{ex:unstabilizable_game}
Consider the matrix game \(\mathcal G(x)=Fx\), where
\[
    F=
    \begin{bmatrix}
        1 & -2 & 1\\
        -2 & 1 & 1\\
        1 & 1 & -2
    \end{bmatrix}.
\]
Since \(F\mathbf 1_3=0\), 
\(
    x^*=\frac{1}{3}\mathbf 1_3
\)
is an interior Nash equilibrium. Moreover, using the basis \(N\) in
\eqref{eq:N_projection}, we obtain
\[
    \widetilde F=N^\top F N
    =
    3
    \begin{bmatrix}
        1 & 0\\
        0 & -1
    \end{bmatrix},
    \qquad
    \widetilde Q=N^\top Q^*N=\frac{1}{3}I_2 .
\]
Hence \(\widetilde Q\widetilde F\) has a positive real eigenvalue \(\lambda=1\). Also, since \(\widetilde F\) is nonsingular, the interior Nash equilibrium \(x^*\) is isolated. Therefore, by Theorem~\ref{thm:positive_eig_unstabilizable}, \(x^*\) cannot be locally asymptotically stabilized by any payoff-based higher-order RD whose payoff filter \(g(s)\) satisfies Assumption~\ref{ass:filter}.
\end{example}

\section{Stabilization Beyond Nash Stationarity}\label{sec:beyond_nash_stationarity}
The preceding result shows that the limitation is not merely a consequence of a particular choice of the payoff filter \(g(s)\). Rather, it is tied to the payoff-to-score transfer-function structure \( G(s)=\frac{1}{s}+g(s), \) which is used by payoff-based higher-order RD to preserve Nash stationarity. The integrator term \(1/s\) is essential for this property, but it also imposes structural constraints on the closed-loop characteristic equation. This raises a natural question: what stability properties can be recovered if Nash stationarity is relaxed? 

We now show that if Nash stationarity is relaxed, then local
stability can be recovered for any continuously differentiable game, but at the
cost of stabilizing a logit equilibrium rather than an exact Nash equilibrium.
These logit equilibria approximate Nash equilibria in an entropy-regularized
sense.

We consider the exponential replicator dynamics (Ex-RD)
\begin{equation} \label{eq:exrd} \dot z=-az+kp, \qquad x=\sigma(z), \end{equation} 
where \(a>0\) and \(k>0\). The transfer function from \(p\) to \(z\) is \( G(s)=\frac{k}{s+a}, \) which is asymptotically stable and strictly proper. The dynamics \eqref{eq:exrd} generalize the exponential replicator dynamics introduced in \cite{gadjov2018passivity}, corresponding to the special case \(a=k=1\). 
Unlike payoff-based higher-order RD, Ex-RD does not satisfy Nash stationarity,
as shown in the following lemma.

\begin{lemma}
\label{lem:exrd_not_NS}
Ex-RD \eqref{eq:exrd} does not satisfy Nash stationarity on
\(\Delta_n\) in general. Moreover, an interior Nash equilibrium
\(x^*\in\operatorname{Int}(\Delta_n)\) is stationary under Ex-RD if and only if
\(
    x^*=\frac{1}{n}\mathbf 1_n .
\)
\end{lemma}

\begin{proof}
First, every strategy generated by Ex-RD satisfies
\(
    x(t)=\sigma(z(t))\in\operatorname{Int}(\Delta_n).
\)
Hence no boundary Nash equilibrium can be stationary under Ex-RD.

Now let \(x^*\in\operatorname{Int}(\Delta_n)\) be a Nash equilibrium, and
suppose that Ex-RD generates a constant strategy trajectory
\(
    x(t)\equiv x^* .
\)
Since \(x^*\) is an interior and Nash equilibrium,
\(
    \mathcal G(x^*)=\alpha\mathbf 1_n
\)
for some \(\alpha\in\mathbb R\). Also, since \(\sigma(z(t))=x^*\), there exists a scalar function \(c(t)\) such that
\(
    z(t)=\log(x^*)+c(t)\mathbf 1_n .
\)
Substituting into Ex-RD gives
\[
    \dot c(t)\mathbf 1_n
    =
    -a\log(x^*)+\big(k\alpha-ac(t)\big)\mathbf 1_n .
\]
Therefore, we must have
\(
    \log(x^*)\in\operatorname{span}\{\mathbf 1_n\}.
\)
Equivalently,
\(
    x^*=\frac{1}{n}\mathbf 1_n .
\)
Thus, any interior Nash equilibrium that is stationary under Ex-RD must be the uniform distribution.

Conversely, suppose
\(
    x^*=\frac{1}{n}\mathbf 1_n
\)
is a Nash equilibrium. Then \(\mathcal G(x^*)=\alpha\mathbf 1_n\) for some
\(\alpha\in\mathbb R\). Choose
\(
    z^*=\frac{k}{a}\alpha\mathbf 1_n .
\)
Then
\(
    -az^*+k\mathcal G(x^*)=0,
\)
and
\(
    \sigma(z^*)=\frac{1}{n}\mathbf 1_n=x^* .
\)
Hence \(x^*\) is stationary under Ex-RD.

Therefore, among interior Nash equilibria, a Nash equilibrium is stationary
under Ex-RD if and only if it is the uniform distribution. Since boundary Nash equilibria are never
stationary under Ex-RD, Ex-RD does not satisfy Nash stationarity on \(\Delta_n\) in general.
\end{proof}

 The next theorem shows that, for any continuously differentiable game, the parameters of Ex-RD \eqref{eq:exrd} can be chosen so that at least one stationary strategy is locally asymptotically stable. The stabilized stationary strategy is a logit equilibrium, and hence need not be a Nash equilibrium.
 
\begin{theorem}
\label{thm:exrd_stable_logit}
Consider a continuously differentiable game
\(\mathcal G:\Delta_n\to\mathbb R^n\). Then there exist \(a>0\) and \(k>0\)
such that the Ex-RD \eqref{eq:exrd}
admits an interior stationary strategy \(x^*\in\operatorname{Int}(\Delta_n)\) that is locally asymptotically stable. Moreover, \(x^*\) is a logit equilibrium:
\(
    x^*=\sigma\!\left(\frac{k}{a}\mathcal G(x^*)\right).
\)
In general, \(x^*\) need not be a Nash equilibrium.
\end{theorem}

\begin{proof}
Let
\(
    \eta:=\frac{k}{a}>0 .
\)
A stationary strategy of Ex-RD satisfies
\(
    0=-az^*+k\mathcal G(x^*),
    \;
    x^*=\sigma(z^*).
\)
Equivalently,
\(
    z^*=\eta\mathcal G(x^*),
\)
and therefore
\(
    x^*=\sigma\!\left(\eta\mathcal G(x^*)\right).
\)
Thus stationary strategies of Ex-RD are precisely logit equilibria with
parameter \(\eta\).  For any fixed \(\eta>0\), define
\(
    T_\eta(x):=\sigma\!\left(\eta\mathcal G(x)\right).
\)
Since \(\mathcal G\) and \(\sigma\) are continuous, \(T_\eta\) is a continuous
map from \(\Delta_n\) into \(\operatorname{Int}(\Delta_n)\subset\Delta_n\).
Hence, by Brouwer's fixed-point theorem, there exists
\(x^*\in\operatorname{Int}(\Delta_n)\) such that
\(
    x^*=T_\eta(x^*).
\)
Setting
\(
    z^*=\eta\mathcal G(x^*)
\)
then gives \(x^*=\sigma(z^*)\) and
\(
    -az^*+k\mathcal G(x^*)=0.
\)
Thus \(x^*\) is a stationary strategy of Ex-RD.

It remains to choose \(\eta\) so that this stationary strategy is locally
asymptotically stable. Linearizing Ex-RD around \(z^*\) gives
\[
    \Delta\dot z
    =
    -a\Delta z
    +
    k\nabla\mathcal G(x^*)\Delta x,
    \qquad
    \Delta x=Q^*\Delta z,
\]
where
\(
    Q^*=\operatorname{diag}(x^*)-x^*{x^*}^{\top}.
\)
Projecting the score linearization onto the tangent space \(\mathcal Z\) gives
\[
    \delta\dot z
    =
    a\left(-I_{n-1}+\eta\widetilde F\widetilde Q\right)\delta z,
\]
where
\(
    \widetilde Q=N^\top Q^*N,
    \;
    \widetilde F=N^\top\nabla\mathcal G(x^*)N .
\)
Let \(\lambda_1,\ldots,\lambda_{n-1}\) denote the eigenvalues of
\(\widetilde F\widetilde Q\). Then the eigenvalues of the reduced score
linearization are
\(
    a(-1+\eta\lambda_i),
    \; i=1,\ldots,n-1 .
\)

For \(x\in\Delta_n\), define
\(
    Q(x):=\operatorname{diag}(x)-xx^\top,
    \;
    \widetilde Q(x):=N^\top Q(x)N,
    \;
    \widetilde F(x):=N^\top\nabla\mathcal G(x)N .
\)

Since \(\mathcal G\) is continuously differentiable and \(\Delta_n\) is compact,
there exists \(L>0\), independent of \(x^*\), such that
\(
    \|\widetilde F(x)\widetilde Q(x)\|_2\le L,
    \;
    \forall x\in\Delta_n .
\)
Then every eigenvalue \(\lambda_i\) of \(\widetilde F\widetilde Q\) satisfies
\(
    \operatorname{Re}(\lambda_i)\le |\lambda_i|\le L .
\)
Choose
\(
    0<\eta<1/L.
\)
Therefore,
\[
    \operatorname{Re}(-1+\eta\lambda_i)
    =
    -1+\eta\operatorname{Re}(\lambda_i)
    \le
    -1+\eta L
    <0 .
\]
Thus
\(
    -I_{n-1}+\eta\widetilde F\widetilde Q
\)
is Hurwitz. Since \(a>0\), the reduced score linearization is asymptotically
stable.

The remaining score direction is also asymptotically stable. Indeed,
because \(Q^*\mathbf 1_n=0\),
\(
    \left(-aI_n+k\nabla\mathcal G(x^*)Q^*\right)\mathbf 1_n
    =
    -a\mathbf 1_n .
\)
Hence the full score linearization is Hurwitz. Therefore, by Lyapunov's indirect method, the stationary strategy \(x^*\) is locally asymptotically
stable. 

Finally, \(x^*\) is a logit equilibrium by construction. In general, it need
not be a Nash equilibrium, since Ex-RD is not Nash stationary by
Lemma~\ref{lem:exrd_not_NS}.
\end{proof}

\begin{remark}
\label{rem:temperature_tradeoff}
The logit equilibrium stabilized in Theorem~\ref{thm:exrd_stable_logit} can be
made closer to satisfying the Nash condition, in the \(\varepsilon\)-Nash sense,
by introducing a softmax temperature. For \(\tau>0\), define
\(
    \sigma_\tau(z)_i
    =
    \exp(z_i/\tau)
    /\sum_{\ell=1}^n \exp(z_\ell/\tau) .
\)
If Ex-RD is implemented with \(x=\sigma_\tau(z)\), then any stationary strategy
satisfies
\(
    x^*
    =
    \sigma_\tau\!\left(\frac{k}{a}\mathcal G(x^*)\right)
    =
    \sigma\!\left(\frac{k}{a\tau}\mathcal G(x^*)\right).
\)
Thus \(x^*\) is a logit equilibrium with effective logit parameter
\(
    \eta=k/a\tau.
\)
Equivalently, \(x^*\) solves the entropy-regularized best-response problem 
\[ x^* \in \arg\max_{x\in\Delta_n} \left\{ x^\top \mathcal G(x^*) - \frac{1}{\eta}\sum_{i=1}^n x_i\log x_i \right\}. \] 
Therefore, for every \(y\in\Delta_n\), 
\[ y^\top\mathcal G(x^*) - {x^*}^{\top}\mathcal G(x^*) \le \frac{1}{\eta} \left( \sum_{i=1}^n y_i\log y_i - \sum_{i=1}^n x_i^*\log x_i^* \right). \] 

Using the entropy bound \( -\log n \le \sum_{i=1}^n x_i\log x_i \le 0, \; \forall x\in\Delta_n, \) we obtain \[ \max_{y\in\Delta_n} y^\top\mathcal G(x^*) - {x^*}^{\top}\mathcal G(x^*) \le \frac{\log n}{\eta} = \frac{a\tau}{k}\log n . \] Thus \(x^*\) is an \(\varepsilon\)-Nash equilibrium with \( \varepsilon=\frac{a\tau}{k}\log n . \) Consequently, decreasing \(\tau\) improves the Nash approximation.

However, the same decrease in \(\tau\) increases the  loop gain
\(\eta=k/(a\tau)\). In the proof of
Theorem~\ref{thm:exrd_stable_logit}, local stability is guaranteed whenever
\(
    \eta<\frac{1}{L},
\)
where \(L>0\) satisfies
\(
    \|\widetilde F(x)\widetilde Q(x)\|_2\le L,
    \;
    \forall x\in\Delta_n .
\)
With temperature \(\tau\), this sufficient stability condition becomes
\(
    \frac{k}{a\tau}<\frac{1}{L},
\)
or equivalently,
\(
    \tau>\frac{kL}{a}.
\)
Consequently, decreasing \(\tau\) improves the \(\varepsilon\)-Nash bound but
can violate the local stability condition. This reveals a tradeoff between
approximate Nash behavior and local closed-loop stability.
\end{remark}

\section{Conclusion}
This paper studied stabilization limits of payoff-based higher-order RD in population games. We first established a necessity result for passivity: if the payoff filter \(g(s)\) is not passive, then there exists a static strictly contractive game whose interior Nash equilibrium is unstable under the resulting closed-loop learning dynamics. Thus, passivity is necessary for universal local stability over contractive games. We then showed that payoff-based higher-order RD has an intrinsic stabilization limitation: there exists a class of games with isolated interior Nash equilibria that cannot be locally asymptotically stabilized by an asymptotically stable strictly proper \(g(s)\). Finally, we showed that this limitation can be bypassed if Nash stationarity is relaxed. In particular, Ex-RD can locally stabilize a logit equilibrium for any continuously differentiable game, although the stabilized equilibrium is generally not a Nash equilibrium. The temperature-dependent formulation further reveals a tradeoff between approximate Nash behavior and local stability. Future work will study robustness of payoff-based higher-order RD in contractive games, including convergence and stability under noisy payoff observations and uncertain game mappings.

\section{References}

\bibliographystyle{IEEEtran}
\bibliography{refs}

\end{document}